\documentclass[journal]{IEEEtran}

\usepackage[letterpaper,margin=0.75in,footskip=0.75in]{geometry}

\makeatletter
\def\endthebibliography{%
  \def\@noitemerr{\@latex@warning{Empty `thebibliography' environment}}%
  \endlist
}
\makeatother

\IEEEoverridecommandlockouts

\usepackage{cite}
\usepackage[breaklinks=true]{hyperref}
\usepackage[cmex10]{amsmath}
\usepackage{amssymb}
\usepackage{amsthm}
\newtheorem{theorem}{Theorem}

\usepackage{mathtools}
\usepackage{commath}
\usepackage{comment}

\usepackage[dvipsnames]{xcolor}

\newtheoremstyle{break}
  {\topsep}{\topsep}%
  {\itshape}{}%
  {\bfseries}{}%
  {\newline}{}%
\theoremstyle{break}

\usepackage{graphicx}

\theoremstyle{definition}
\newtheorem{definition}{Definition}

\DeclareMathOperator*{\argmax}{arg\,max}

\usepackage{tikz}
\usetikzlibrary{shapes}

\title{\LARGE \bf
A Two-Stage Mechanism for Demand Response Markets}

\author{\IEEEauthorblockN{Bharadwaj Satchidanandan, Mardavij Roozbehani, and Munther A. Dahleh}

	\IEEEauthorblockA{Laboratory for Information and Decision Systems\\ Massachusetts Institute of Technology}	
}

\begin{document}
\newgeometry{left=0.75in,right=0.75in,bottom=0.75in,top=1in}
\maketitle
\thispagestyle{empty}
\pagestyle{empty}

\begin{abstract}
Demand response involves system operators using incentives to modulate electricity consumption during peak hours or when faced with an incidental supply shortage. However, system operators typically have imperfect information about their customers' baselines, that is, their consumption had the incentive been absent. The standard approach to estimate the reduction in a customer's electricity consumption then is to estimate their counterfactual baseline. However, this approach is not robust to estimation errors or strategic exploitation by the customers and can potentially lead to overpayments to customers who do not reduce their consumption and underpayments to those who do. Moreover, optimal power consumption reductions of the customers depend on the costs that they incur for curtailing consumption, which in general are private knowledge of the customers, and which they could strategically misreport in an effort to improve their own utilities even if it deteriorates the overall system cost. The two-stage mechanism proposed in this paper circumvents the aforementioned issues. In the day-ahead market, the participating loads are required to submit only a probabilistic description of their next-day consumption and costs to the system operator for day-ahead planning. It is only in real-time, if and when called upon for demand response, that the loads are required to report their baselines and costs. They receive credits for reductions below their reported baselines. The mechanism for calculating the credits guarantees incentive compatibility of truthful reporting of the probability distribution in the day-ahead market and truthful reporting of the baseline and cost in real-time. The mechanism can be viewed as an extension of the celebrated Vickrey-Clarke-Groves mechanism augmented with a carefully crafted second-stage penalty for deviations from the day-ahead bids.

\end{abstract}


\section{Introduction}\label{introduction}
{D}{emand} Response (DR) refers to a variety of mechanisms that aim to actively engage otherwise passive consumers, with the aim to modulate their electricity consumption during peak hours \cite{jordehi2019, deng2015survey, Vardakas2015}. It has been estimated that in some systems, over 20\% of energy costs are driven by just 2\% of peak hours. Thus, targeted demand reductions during particular periods can provide substantial value. In addition to improving system-wide economic efficiency, demand response can also provide ancillary services, defer network capacity upgrade costs, and match demand with intermittent supply from renewable resources \cite{siano2014demand, pinson2014benefits, lee2013assessment, Kwon2017}. It is expected that the value of demand response will continue to grow as the share of renewable energy in the generation portfolio increases. 

The literature addressing the challenges associated with the design of demand response mechanisms is broad, see for example \cite{deng2015survey, Wang2018, Xu2019, jordehi2019} for recent reviews. The most common form of demand response involves voluntary load reduction by participating consumers in exchange for monetary compensation. These programs are often referred to as incentive-based DR, price-based DR, or imputed DR \cite{hogan2010demand}, and exclude mechanisms such as direct load control. It is estimated that more than 2.7 million customers in the United States are enrolled in incentive-based demand response \cite{SEPA}.

Incentive-based DR programs share a common challenge: How to determine the reduction in a customer's consumption during a demand response period? This involves counterfactual estimation of consumer baselines \cite{varaiyaDR1} and is prone to estimation errors and strategic exploitation. Errors in baseline estimation can lead to overpayments to customers who do not actually reduce their consumption, and underpayments to those who do \cite{varaiyaDR2}. See for example \cite{schneider2021online} and references therein for a detailed discussion of the literature on baseline estimation. 

Moreover, optimal consumption reductions of the customers depend on the costs that they incur for curtailing consumption. These costs are in general private to the customers, and they could strategically misreport them in an effort to improve their own utilities even if it results in an increase in the overall system cost. 

In this paper, we present a mechanism that circumvents the issues that arise from baseline estimation errors and private costs by incentivizing demand response providers to truthfully self-report their baselines and costs. It is based on the mechanism developed in \cite{arxivAER} for two-stage repeated stochastic games. In the day-ahead market, participating consumers may not know exactly their baselines and costs the next day. Rather, they may only know a probabilistic description thereof. Hence, the mechanism allows for the demand response providers to submit only a \emph{probabilistic description} of their next-day consumption and costs to the system operator for day-ahead planning. It is only in real time, if and when called upon for demand response, that the customers are required to report their baselines and costs. Based on the real-time reports, optimal demand reductions are computed for all customers, and they receive credits for reductions below their reported baselines. Our mechanism for calculating the credits ensures ex post incentive compatibility of truthful bidding of the probability distributions in the day-ahead market and truthful reporting of the realized baselines and costs in real time. The mechanism can be viewed as an extension of the celebrated Vickrey-Clarke-Groves mechanism augmented with a carefully crafted second-stage penalty for deviations from the day-ahead bids.

{While the mechanism itself is adapted from \cite{arxivAER}, the main contribution of the paper lies in modeling the interaction between demand response providers and the system operator as a two-stage repeated stochastic game, thereby setting the stage for the application of the results developed in \cite{arxivAER} to design an efficient and incentive-compatible demand response market. }

The rest of the paper is organized as follows. Section \ref{section_relatedWork} provides an account of related work and highlights the contributions of this paper to the state of the art. Section \ref{problemFormulation} formulates the mechanism design problem. Section \ref{mechanism} presents the two-stage mechanism for demand response markets and establishes the incentive and optimality guarantees that it provides. Section \ref{numerical} presents certain numerical results that illustrate the cost benefits of the proposed mechanism as compared to a popular alternative. Section \ref{conclusion} concludes the paper.

\noindent\textbf{Notation:} Given a sequence $\{x(1),x(2),\hdots\},$ we denote by $x^l$ the sequence $\{x(1),\hdots,x(l)\}$ and by $x^\infty$ the entire sequence. Given a vector $\mathbf{x}$ of $n$ components, we denote by ${x}_{i}$ its $i$th component and by $\mathbf{x}_{-i}$ the vector of all other components. Finally, $1_{\{\cdot\}}$ is used to denote the indicator function. The paper is rather notation-heavy, and we have collected the list of symbols in the appendix for quick reference. 

\section{Related Work}\label{section_relatedWork}

A common approach to estimating the demand reduction of DR providers is to estimate their counterfactual baseline. A data-driven method to estimate baseline consumption of residential buildings is presented in \cite{dataDrivenBaselineEstimation}. A model based on features such as temperature and time of the week is presented in \cite{CallawayCounterfactual} to model the baseline consumption of a load, and a regression-based method to compute the baseline model parameters is presented. Also presented in \cite{CallawayCounterfactual} is a cross validation-based method to determine the error magnitudes in baseline parameter estimation. In the absence of carefully designed mechanisms, baseline manipulation can improve a demand response provider's revenue, and \cite{wang2018overconsume} derives an optimal baseline reporting strategy for a demand response provider. Reference \cite{DRexchange} proposes a demand response exchange for trading demand response capacity among DR providers and DR buyers. However, it does not account for baseline stochasticity, baseline inflation, strategic bidding, etc. 

Prior works that are the closest to our paper are \cite{KalathilShakkottai, VaraiyaDR3, varaiyaDR1, varaiyaDR2} which present mechanisms for demand response markets that incentivize demand response providers to truthfully report their baselines. While there are many differences between the formulations in \cite{KalathilShakkottai, VaraiyaDR3, varaiyaDR1, varaiyaDR2} and this paper, we highlight three major ones here. First, the aforementioned papers restrict attention to piecewise linear valuations for demand response providers and leave the case of general concave valuations as an open problem. On the other hand, the results of this paper not only apply to concave valuations, but also to arbitrary valuation functions. Secondly and more importantly, all of these papers assume the baselines and costs of DR providers to be deterministic so that they are perfectly known in the day-ahead market. 
However, an important characteristic of real-world loads is that their baselines and costs are in general random and not perfectly known in the day-ahead market. Motivated by this, we develop a mechanism that requires the DR providers to only report a \emph{probabilistic description} of their baselines and costs in the day-ahead market, and report their actual realizations only in real time. The mechanism incentivizes demand response providers to truthfully report the probability distribution \emph{and} the realization of these quantities in the day-ahead market and in real time respectively. Finally, the mechanisms developed in \cite{KalathilShakkottai, VaraiyaDR3} implement truth-telling only as a Nash equilibrium whereas the mechanism presented in this paper implements truth-telling in a stronger notion of equilibrium, namely Dominant Strategy Non-Bankrupting equilibrium. References \cite{varaiyaDR1, varaiyaDR2} implement truth-telling in dominant strategies but under the assumptions of deterministic piecewise linear costs and deterministic baselines. Since our formulation does not assume any specific form for the cost functions, nor does it assume the baselines or costs to be deterministic, it is considerably more general than the formulations in \cite{KalathilShakkottai, VaraiyaDR3, varaiyaDR1, varaiyaDR2}. References \cite{MuthirayanDR, DROptions, related1, related2, related3} are a few other papers addressing mechanism design for DR. 

\section{Problem Formulation}\label{problemFormulation}

Fig. \ref{fig_sketch} illustrates the overall setup, the individual components of which are elaborated in the ensuing subsections. The interaction between the Independent System Operator (ISO) and DR providers occurs over two stages, namely, the day-ahead market and the real-time market. Optimal operation of the grid involves the ISO solving the Economic Dispatch (ED) problem in the day-ahead market. In a power system containing DR providers, the ED problem must take into account the costs incurred by the DR providers in addition to the costs incurred by the generators and the reserves. In Sections \ref{Sec_entities} thru \ref{Sec_DRPolicy}, we introduce the quantities that are necessary to formulate the ED problem with DR providers. In Section \ref{SecDAMarket}, we formulate the ED problem with DR providers. The solution to the problem, as we will see, is a function of certain probability distributions that are privately known to the DR providers (denoted by $\theta_is$ in Fig. \ref{fig_sketch}). Hence, a DR provider could potentially misreport its probability distribution to cause the ISO to take a day-ahead market decision that improves its own utility even if it deteriorates the overall system cost. Optimal grid operation also requires the ISO to take certain recourse actions in real time as a function of certain random variables that realize privately to each DR provider in real time (denoted by $\delta_is$ in Fig. \ref{fig_sketch}). The DR providers could misreport these quantities too in real time in an effort to improve their own respective utilities. These issues motivate the mechanism design problem which we formulate in Section \ref{sec_mechanismDesign}.

\begin{figure}
\centering

\begin{tikzpicture}[scale=0.9]

\node[inner sep=0pt] (EV) at (1,9.9)
    {\includegraphics[width=.06\textwidth]{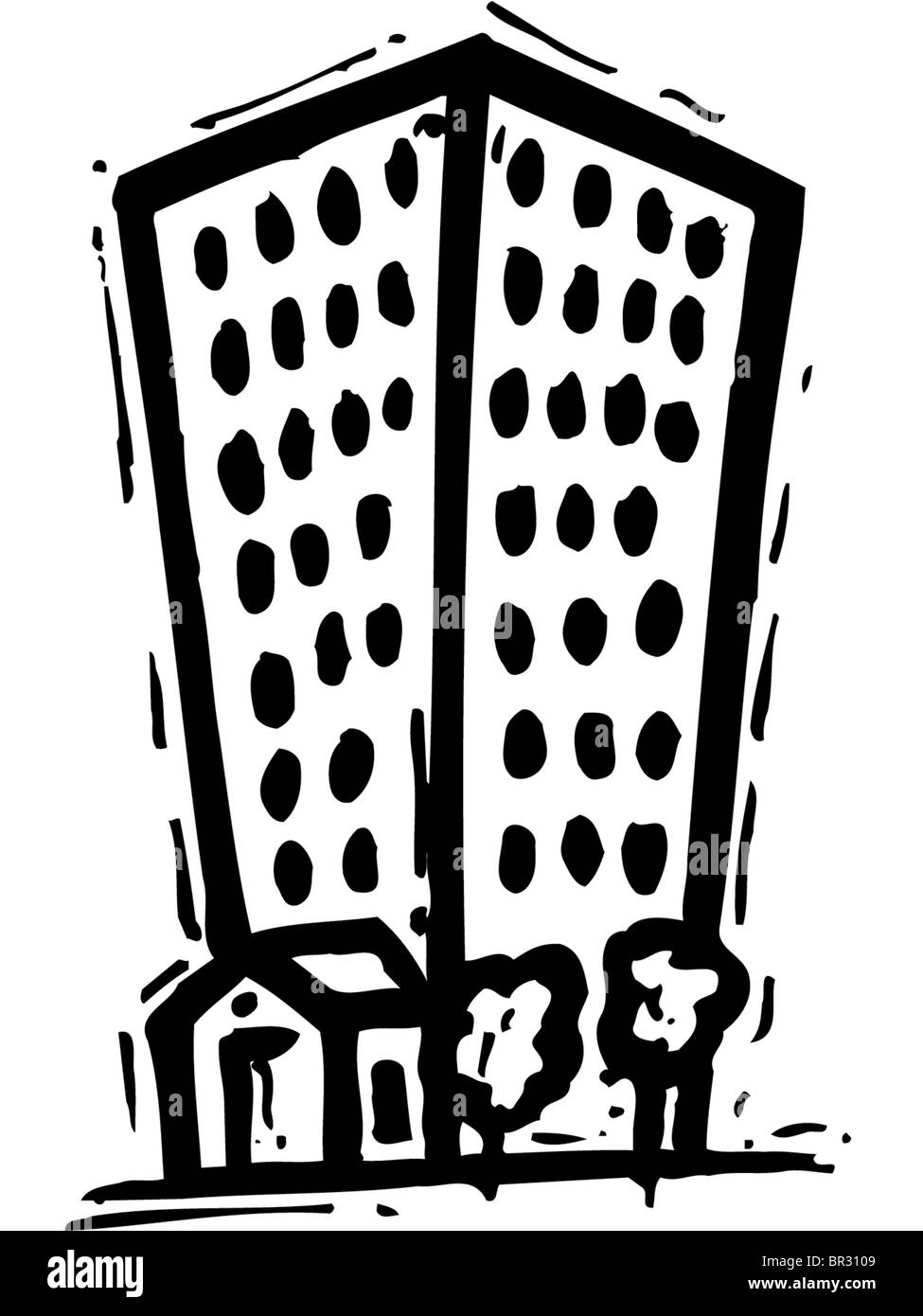}};

\node[inner sep=0pt] (EV) at (3,9.5)
    {\includegraphics[width=.06\textwidth]{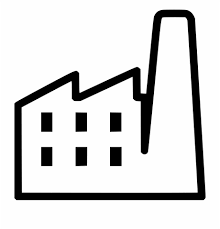}};

\node[align=center] at (4.5,9.5) {$\boldsymbol{\hdots}\;\boldsymbol{\hdots}$};    
    
\node[inner sep=0pt] (EV) at (6,9.6)
    {\includegraphics[width=.06\textwidth]{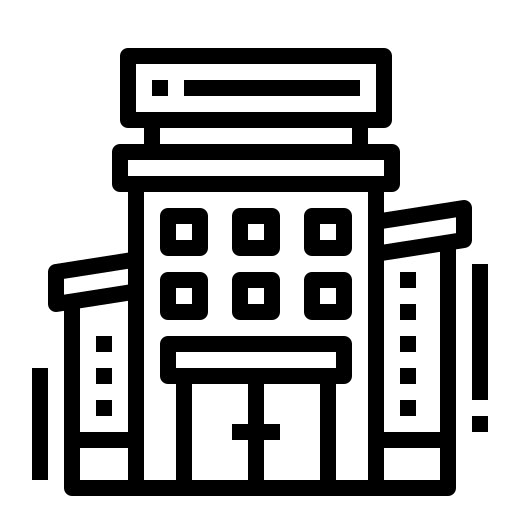}};

\node[align=center] at (7.5,9.5) {$\boldsymbol{\hdots}\;\boldsymbol{\hdots}$};

\node[inner sep=0pt] (EV) at (9,9.9)
    {\includegraphics[width=.06\textwidth]{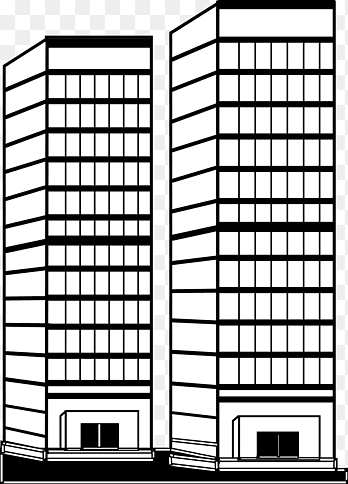}};

\node[align=center] at (1,8.8) {\footnotesize{DR load $1$}};
\node[align=center] at (3,8.8) {\footnotesize{DR load $2$}};
\node[align=center] at (6,8.8) {\footnotesize{DR load $i$}};
\node[align=center] at (9,8.8) {\footnotesize{DR load $n$}};

\node[align=center] at (1,8.3) { \textcolor{ForestGreen}{\footnotesize$\theta_1$}};
\node[align=center] at (3,8.3) { \textcolor{ForestGreen}{\footnotesize$\theta_2$}};
\node[align=center] at (6,8.3) { \textcolor{ForestGreen}{\footnotesize$\theta_i$}};
\node[align=center] at (9,8.3) { \textcolor{ForestGreen}{\footnotesize$\theta_n$}};

\draw[black, thick, ->] (1,8) -> (4.7,7);
\draw[black, thick, ->] (3,8) -> (4.9,7);
\draw[black, thick, ->] (6,8) -> (5.1,7);
\draw[black, thick, ->] (9,8) -> (5.3,7);

\node[inner sep=0pt] (SP) at (5,6.5)
    {\includegraphics[width=.075\textwidth]{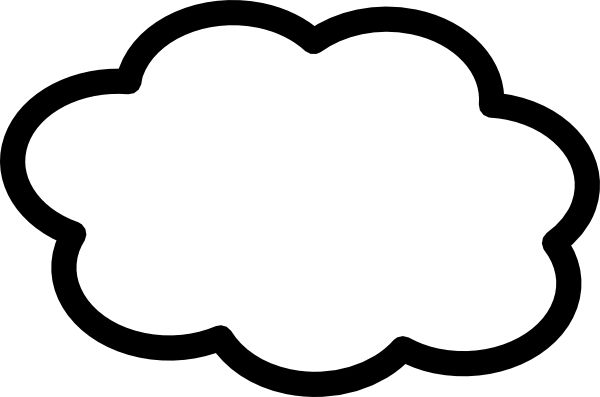}};
    
\node[align=center] at (5,6.5) {\footnotesize{{ISO}}};
\node[align=center] at (2.5,6.5) {\footnotesize{{Day-Ahead Market}}};

\draw[black, thick, ->] (5,6) -- (5,5.8);
\node[align=center] at (5,5.7) { \footnotesize$\big(g^*_{\textcolor{red}{\widehat{{\theta}}_\otimes}},\boldsymbol{\pi}^*_{\textcolor{red}{{\widehat{\theta}_\otimes}}}\big)$};

\node[align=center] at (1.8,7.5) { \textcolor{red}{\footnotesize$\widehat{\theta}_1$}};
\node[align=center] at (4.4,7.5) { \textcolor{red}{\footnotesize$\widehat{\theta}_2$}};
\node[align=center] at (5.9,7.5) { \textcolor{red}{\footnotesize$\widehat{\theta}_i$}};
\node[align=center] at (8.2,7.5) { \textcolor{red}{\footnotesize$\widehat{\theta}_n$}};

\draw[black, dotted, thick] (5,5.5) -- (5,5.3);
\draw[black, dotted, thick] (1,5.3) -- (9,5.3);
\draw[black, dotted, thick, ->] (1,5.3) -- (1,5.15);
\draw[black, dotted, thick, ->] (3,5.3) -- (3,4.6);
\draw[black, dotted, thick, ->] (6,5.3) -- (6,4.6);
\draw[black, dotted, thick, ->] (9,5.3) -- (9,5.15);

\node[inner sep=0pt] (EV) at (1,4.3)
    {\includegraphics[width=.06\textwidth]{apartment.jpg}};

\node[inner sep=0pt] (EV) at (3,3.9)
    {\includegraphics[width=.06\textwidth]{factory.png}};

\node[align=center] at (4.5,4.3) {$\boldsymbol{\hdots}\;\boldsymbol{\hdots}$};    
    
\node[inner sep=0pt] (EV) at (6,4)
    {\includegraphics[width=.06\textwidth]{mall.png}};

\node[align=center] at (7.5,4.3) {$\boldsymbol{\hdots}\;\boldsymbol{\hdots}$};

\node[inner sep=0pt] (EV) at (9,4.3)
    {\includegraphics[width=.06\textwidth]{apartment2.png}};

\node[align=center] at (1,3.2) { \textcolor{ForestGreen}{\footnotesize$\delta_1(l)\sim\theta_1$}};
\node[align=center] at (3,3.2) { \textcolor{ForestGreen}{\footnotesize$\delta_2(l)\sim\theta_2$}};
\node[align=center] at (6,3.2) { \textcolor{ForestGreen}{\footnotesize$\delta_i(l)\sim\theta_i$}};
\node[align=center] at (9,3.2) { \textcolor{ForestGreen}{\footnotesize$\delta_n(l)\sim\theta_n$}};

\draw[black, thick, ->] (1,3) -- (4.7,2);
\draw[black, thick, ->] (3,3) -> (4.9,2);
\draw[black, thick, ->] (6,3) -> (5.1,2);
\draw[black, thick, ->] (9,3) -> (5.3,2);

\node[align=center] at (1.8,2.5) { \textcolor{red}{\footnotesize${\widehat{\delta}_1}(l)$}};
\node[align=center] at (4.4,2.6) { \textcolor{red}{\footnotesize$\widehat{\delta}_2$(l)}};
\node[align=center] at (6.1,2.6) { \textcolor{red}{\footnotesize$\widehat{\delta}_i$(l)}};
\node[align=center] at (8.2,2.5) { \textcolor{red}{\footnotesize$\widehat{\delta}_n$(l)}};

\node[inner sep=0pt] (SP) at (5,1.5)
    {\includegraphics[width=.075\textwidth]{cloud.jpg}};
    
\node[align=center] at (5,1.5) {\footnotesize{{ISO}}};
\node[align=center] at (2.5,1.5) {\footnotesize{{Real-time Market}}};

\draw[black, thick, ->] (5,1) -- (5,0.5);
\node[align=center] at (5,0.2) { \footnotesize$\boldsymbol{\pi}^*_{\textcolor{red}{\widehat{\boldsymbol{\theta}}_\otimes}}(z(l),\textcolor{red}{\widehat{\boldsymbol{\delta}}(l)})$};

\end{tikzpicture}
\caption{The DR loads observe their type distribution $\{\theta_i\}_{i=1}^n$ in the day-ahead market. They report $\{\widehat{\theta}_i\}_{i=1}^n$ to the ISO. The ISO takes the day-ahead market decisions $\mathbf{g}^*_{\widehat{\theta}_\otimes}$ and $\boldsymbol{\pi}^*_{\widehat{\theta}_\otimes}$ based on the reported distributions. The DR loads observe the decisions. In real-time on any day $l$, the types $\{\delta_i(l)\}_{i=1}^n$ of the DR loads realize. They report $\{\widehat{\delta}_i(l)\}_{i=1}^n$ as their types. The ISO computes the real-time decisions on day $l$ based on the net demand and the reported types on day $l$.}\label{fig_sketch}
\end{figure}

\subsection{The entities in the power system}\label{Sec_entities}

Consider a power system consisting of (i) an ISO, (ii) a generator that should be scheduled in the day-ahead market, (iii) a fast-acting reserve from which the ISO can purchase energy in the spot market to balance real-time demand-supply mismatches, (iv) a renewable energy source such as solar or wind, (v) $n$ demand response-providing loads which we refer to henceforth as the DR loads, and (vi) inelastic loads.

We divide time into days and each day into several time intervals. {There are often constraints such as ramp rates that couple energy generation across different time intervals of a day, and constraints such as total energy requirements (eg., charge requirements of electric vehicles) that couple the demand across different time intervals of a day. However, for ease of exposition of the main ideas, we assume that there are no inter-temporal constraints that couple generation or demand across time so that it suffices to describe the problem and the mechanism for an arbitrarily fixed time interval $t$.} Hence, in all forthcoming notation, we only include the index $l\in\mathbb{N}$ denoting the day and suppress the index $t$ which denotes the time interval within the day.

\subsection{Demand of inelastic loads and DR providers}\label{Sec_demand}
Denote by $z(l)$ the net demand of inelastic loads on day $l\in\mathbb{N}$ (at the arbitrarily fixed time $t$), by which we mean the difference between the sum demand of all inelastic loads and the renewable energy that is produced on day $l$ (at the arbitrarily fixed time $t$). A quintessential feature of renewable energy sources is that the energy that they produce on any day $l$ is a random variable in the day-ahead market. Consequently, $z(l)$ is a random variable in the day-ahead market and we denote by $\theta_z$ its probability distribution. We also assume that $\{z(1),z(2),\hdots,\}$ is an Independent and Identically Distributed (IID) sequence.

Denote by $d_{i}(l)$ the \emph{baseline demand} of DR load $i$ on day $l$, $i\in\{1,\hdots,n\}$ and $l\in\mathbb{N}.$ By baseline demand $d_{i}(l)$, we mean the energy that DR load $i$ consumes on day $l$ if it does not adjust its energy consumption for demand response. To keep the analysis simple, we quantize power so that $d_{i}(l)\in\{0,1,\hdots,d_{\mathrm{max}}\}$ for all $i,l.$ 

\subsection{Types}
Real-world loads exhibit two important characteristics: (i) Their baselines could be different on different days, and (ii) their appetite for reducing energy consumption could be different on different days. We model these two factors by the notion of a \emph{type}. Hence, the baseline of a DR load and its appetite for reducing energy consumption being different on different days is tantamount to its type being different on different days. Section \ref{sec_costs} provides a precise operational meaning for the type of a DR load. We denote by $\Delta$ the type space of the DR loads. 
 
We model the time variation of a DR load $i$'s type by modeling it as a random variable that is drawn independently on each day $l$ from a probability distribution $\theta_i.$ Hence, denoting by $\delta_i(l)\in\Delta$ the type of DR load $i$ on day $l$, the sequence $\{\boldsymbol{\delta}(1),\boldsymbol{\delta}(2),\hdots\}$ is an IID sequence of random variables with $\boldsymbol{\delta}(1)\sim\theta_1\times\hdots\times\theta_n=:\theta_\otimes.$ 

Both the type distribution $\theta_i$ and the type realization $\{\delta_i(1),\delta_i(2),\hdots\}$ of DR load $i$ are private knowledge of the load. Moreover, these quantities are revealed to the load at different times. Specifically, in the day-ahead market, the load does not know precisely its baseline the next day, nor does it know precisely its appetite for reducing consumption the next day. Rather, it only knows the probability distribution $\theta_i$ according to which its type is drawn. The load observes its type realization $\delta_i(l)$ only on day $l$. 

Note that we have assumed the type distribution of each DR load to remain the same on all days. This is only for ease of exposition and to simplify to some extent the already complicated notation. With minor modifications, the results of this paper will continue to hold even in the case when the type distributions are different on different days but the number of different distributions in $L$ days is $o(L).$

\subsection{Cost functions}\label{sec_costs}

Providing demand response involves the DR loads reducing their energy consumption from their baseline to a lower quantity during times of power supply shortage. Power supply shortages could occur, for example, due to low renewable energy generation. Power consumption reductions come at a certain discomfort or productivity loss to the DR loads which we capture by means of a cost function. The cost incurred by a DR load on any given day is a function of both the consumption reduction and its type on that particular day. Specifically, we denote by $c_i(x,\delta_i(l))\in\mathbb{R}$ the cost incurred by DR load $i$ on day $l$ for consuming $x$ units of energy below its baseline (recall that the type $\delta_i(l)$ specifies the DR load's baseline on day $l$). 

Generators that have low ramp rates must have their power dispatch scheduled well in advance of the time of power delivery, and this is typically done in the day-ahead market. We denote by $c_g:\mathbb{R}_+\to\mathbb{R}$ the production function of the generator so that $c_g(x)$ specifies the cost that the generator incurs for producing $x$ units of energy. 

We assume the availability of a fast-acting reserve generator that can produce or consume energy in real time to balance real-time demand-supply mismatches. 
We denote by $c_r:\mathbb{R}\to\mathbb{R}$ the production function of the reserve so that $c_r(x)$ specifies the cost incurred by the reserve generator for producing $x$ units of energy.

\subsection{Demand response policy}\label{Sec_DRPolicy}


On any day $l$, the ISO observes the supply shortage $z(l)$ and each DR load $i$ observes its type ${\delta}_i(l)$. The type realizations are private to the DR loads but the ISO requires their knowledge to determine the optimal consumption reduction of each DR load on day $l$. Hence, the ISO requests the DR loads to report their types, but for strategic reasons that will become clear in Section \ref{sec_mechanismDesign}, they may not bid their types truthfully. Hence, we denote by $\widehat{\delta}_i(l)$ the type reported by DR load $i$ on day $l.$ A \emph{demand response policy} $\boldsymbol{\pi}:\mathbb{R}\times\Delta^n\to\mathbb{R}^n$ is a rule by which the ISO determines the amount by which each DR load should reduce its energy consumption from its baseline as a function of $z(l)$ and the reported types ${\widehat{\boldsymbol{\delta}}}(l)$. 

\subsection{The day-ahead market}\label{SecDAMarket}
There are two decisions that the ISO should compute in the day-ahead market: 
\begin{enumerate}
    \item[(i)] the power dispatch of the generator, and 
    \item[(ii)] the demand response policy of the DR loads. 
\end{enumerate}
These quantities must be computed only based on the information that is available in the day-ahead market, namely, the type distributions $\theta_1,\hdots,\theta_n$ of the DR loads, which are known privately to them, and the distribution $\theta_z$ of the net demand of the inelastic loads, which is public knowledge.

Since the type distributions are private to the DR loads, the ISO requests the DR loads to report them in the day-ahead market. However, for strategic reasons that we describe in Section \ref{sec_mechanismDesign}, the DR loads may not bid their type distributions truthfully and so we denote by $\widehat{\theta}_i$ the type distribution reported by DR load $i.$

Suppose for a moment that the DR loads bid their type distributions truthfully so that $\widehat{\theta}_{\otimes}\coloneqq\widehat{\theta}_1\times\hdots\times\widehat{\theta}_n=\theta_\otimes$, and that they bid their type realizations truthfully in real time so that $\boldsymbol{\widehat{\delta}}(l)=\boldsymbol{\delta}(l)$ for all $l$. How should the ISO compute the power dispatch of the generator and the demand response policy of the DR loads in this case?

To answer this, suppose that energy dispatch of the generator is chosen to be $g$ and that the demand response policy is chosen to be $\boldsymbol{\pi}$. Then, the demand-supply mismatch that occurs in real time on any day $l$ is $$z(l)-g+\sum_{i=1}^n\big[\widehat{d}_i(l)-\pi_i\big(z(l),\widehat{\boldsymbol{\delta}}(l)\big)\big]=:g_r(l)$$ where $\pi_i(z(l),\widehat{\boldsymbol{\delta}}(l))$ denotes the $i$th component of $\boldsymbol{\pi}(z(l),\widehat{\boldsymbol{\delta}}(l))$ and $\widehat{d}_i(l)$ denotes the baseline reported by DR load $i$ on day $l$ (which, recall, is specified by $\widehat{\delta}_i(l)$). The ISO must purchase $g_r(l)$ units of energy from the reserve at cost $c_r(g_r(l))$. Note that $g_r(l)$ is a random variable in the day-ahead market whose distribution depends on $\theta_\otimes,$ which we have assumed for the moment is known to the ISO. Consequently, the expected social cost --- defined as the total cost incurred by the generator, the reserve, and the DR loads --- on day $l$ is 
\begin{align}
    W(g&,\boldsymbol{\pi},\widehat{\theta}_\otimes)=\mathbb{E}_{(z(l),\widehat{\boldsymbol{\delta}}(l))\sim\theta_z\times\widehat{\theta}_\otimes}\bigg[c_g(g)+c_r(g_r(l))\nonumber\\
    &\;\;\;\;\;\;\;\;\;\;\;\;\;\;\;\;\;\;\;\;\;\;\;\;\;\;\;+\sum_{i=1}^nc_i\big(\pi_i(z(l),\widehat{\boldsymbol{\delta}}(l)),\widehat{\delta_i}(l)\big)\bigg].
\end{align}
The goal of the ISO in the day-ahead market is to minimize the total expected social cost the following day and therefore determines the power dispatch $g^*_{\widehat{\theta}_\otimes}$ and demand response policy $\boldsymbol{\pi}^*_{\widehat{\theta}_\otimes}$ as
\begin{align}
    (g^*_{\widehat{\theta}_\otimes},\boldsymbol{\pi}^*_{\widehat{\theta}_\otimes})=\underset{g,\boldsymbol{\pi}}{\mathrm{argmin}}\;\;W(g,\boldsymbol{\pi},\widehat{\theta}_\otimes).\label{DAdecision}
\end{align}
We denote by $W^*(\widehat{\theta}_\otimes)$ the optimal value of the above objective. 

The problem of course is that the reported type distributions $(\widehat{\theta}_1,\hdots,\widehat{\theta}_n)$ may not be equal to the true type distributions $(\theta_1,\hdots,\theta_n)$, and so the decision $(g^*_{\widehat{\theta}_\otimes},\pi^*_{\widehat{\theta}_{\otimes}})$ which the ISO computes in the day-ahead market may not be equal to the optimal decision $(g^*_{\theta_\otimes},\pi^*_{\theta_\otimes})$. We describe in Section \ref{mechanism} the mechanism by which the ISO can elicit the type distributions truthfully.

{It is worth noting that the optimal power dispatch and demand response policy are functions of only $\theta_\otimes$ and $\theta_z$ which we assume remain the same on all days. Hence, it suffices for the DR loads to report their type distributions just once and for the ISO to compute the optimal power dispatch and demand response policy just once, namely, in the day-ahead market before day $1$. It can reuse these decisions on all days without any loss of optimality.} As mentioned before, with minor modifications, the results of this paper extend to a more general case wherein the distributions are different on different days but the number of different distributions in $L$ days is $o(L).$



\subsection{The real-time market}\label{SecRTMarket}
After the day-ahead decisions are made, the net demand $z(l)$ of the inelastic loads and the type profile $\boldsymbol{\delta}(l)$ of the DR loads realize on day $l$. The ISO requests the DR loads to report their type realizations and we denote by $\widehat{\delta}_i(l)$ the type reported by DR load $i$ on day $l,$ which may or may not be equal to its true type $\delta_i(l).$ The ISO then computes $\boldsymbol{\pi}^*_{\widehat{\theta}_\otimes}(z(l),\boldsymbol{\widehat{\delta}}(l))$, informs each DR load the amount by which it should reduce its consumption, and purchases the residual mismatch from the spot market. Each DR load $i$ is contractually obligated to set its actual energy consumption $y_i(l)$ as $$y_i(l)=\widehat{d}_i(l)-\pi^{*}_{\widehat{\theta}_\otimes,i}(z(l),\widehat{\boldsymbol{\delta}}(l)),$$
where $\pi^{*}_{\widehat{\theta}_\otimes,i}(z(l),\widehat{\boldsymbol{\delta}}(l))$ denotes the $i$th component of $\boldsymbol{\pi}^*_{\widehat{\theta}_\otimes}(z(l),\boldsymbol{\widehat{\delta}}(l)).$ Note that the ISO can check for each DR load $i$ whether its actual energy consumption satisfies the above equality, and declare it to be non-compliant with the DR program if not.

Note also that the ISO conducting the above check does not imply that a DR load $i$ truthfully reduces its energy consumption by $\pi_{\widehat{\theta}_\otimes,i}^*(z(l),\boldsymbol{\widehat{\delta}}(l))$ since it could misreport its baseline. For example, it could inflate its baseline by reporting $\widehat{d}_i(l)>d_i(l)$ to give the impression of curtailing consumption without actually doing so. The mechanism presented in Section \ref{mechanism} incentivizes the DR loads to truthfully report their baselines, thereby obligating them to reduce their real-time consumption by the ISO-specified amounts.

\subsection{Bidding strategies of DR loads}
As described in Section \ref{SecDAMarket}, the DR loads only know their type distributions in the day-ahead market which they have to report to the ISO. We denote by $\sigma_i:\Theta\to\Theta$ the day-ahead bidding strategy of DR load $i$ so that $\widehat{\theta}_i=\sigma_i(\theta_i)$. Here, $\Theta$ denotes the set of probability distributions on $\Delta$. 

In real time, DR loads bid their type realizations. We allow for the type bid $\widehat{\delta}_i(l)$ of DR load $i$ to be constructed based on all information that is available to it until day $l$, and in accordance with any arbitrary, randomized, history-dependent policy. 
Hence, a real-time bidding policy $\phi$ of DR load $i$ specifies $\mathbb{P}_{\phi}(\widehat{\delta}_i(l)\vert\delta_i^l,{(\pi^*_{\widehat{\theta}_\otimes,i})}^{l-1},g^*_{\widehat{\theta}_\otimes})$ for each $l\in\mathbb{N}$. I.e., it specifies a probability distribution over the type space $\Delta$ according to which $\widehat{\delta}_i(l)$ is chosen as a function of all observations available to DR load $i$ until day $l$. We denote by $\Phi_i$ the set of all real-time bidding policies.

Observe that the real-time bidding policy is a rule which specifies how a DR load should construct its type bid on any given day. While the output of the rule on any given day is a random variable which depends on the realization of the type sequence $\boldsymbol{\delta}^\infty,$ \emph{there is nothing random about the rule itself}. A load without any loss of generality can choose the rule in the day-ahead market corresponding to day $1$, and as a function of $\theta_i$ --- the only information that is available to it at that time. This observation leads to the notion of a \emph{real-time bidding strategy}. A real-time bidding strategy of DR load $i$ is a function that maps its type distribution $\theta_i$ to a real-time bidding policy in $\Phi_i.$ We denote by $\mu_i:\Theta\to\Phi_i$ the real-time bidding strategy of DR load $i$ so that $\mu_i(\theta_i)$ is its real-time bidding policy. 

Note that both the day-ahead bidding strategy $\sigma_i$ and the real-time bidding strategy $\mu_i$ are functions on the set $\Theta$, and a DR load $i$ without any loss of generality can choose these functions ``offline," i.e., even before it observes $\theta_i$.

We refer to the combination $S_i\coloneqq(\sigma_i,\mu_i)$ as the \emph{strategy} of DR load $i$ and denote by $\mathcal{S}_i$ the set of strategies available to DR load $i.$ Note that once all DR loads fix their strategies, a functional relationship is established between $\boldsymbol{\widehat{\delta}}^\infty$ and $\boldsymbol{\delta}^\infty,$ and all random variables become well defined. 

\begin{definition}
A strategy $(\sigma,\mu)$ of DR load $i,$ $i\in\{1,\hdots,n\},$ is \emph{truthful} if 
\begin{enumerate}
    \item $\sigma(\theta)=\theta$ for all $\theta\in\Theta$, and
    \item there exists $\mathcal{L}\subset\mathbb{N}$ with $\sum_{k=1}^L{\bf 1}_{\{k\in\mathcal{L}\}}=o(L)$ such that for all $l\notin\mathcal{L},$ $$\mathbb{P}_{\mu}(\widehat{\delta}_i(l)\big\vert\delta_i^l,({\pi_{\widehat{\theta}_\otimes,i}^*})^{l-1},g^*_{\widehat{\theta}_\otimes})={\bf 1}_{\{\widehat{\delta}_i(l)=\delta_i(l)\}}.$$
\end{enumerate}
In words, a truthful strategy reports the type distribution truthfully and reports the type realization truthfully ``almost all days."
\end{definition}
We denote by $\mathcal{T}_i$ the set of truthful strategies available to DR load $i.$

\subsection{Payments and Utilities}\label{Sec_utilities}
The ISO has for every $(i,l)\in\{1,\hdots,n\}\times\mathbb{N}$ a payment rule $p_{i,l}:\theta_\otimes\times\Delta^{n\times l}\to\mathbb{R}$ that determines the payment that DR load $i$ receives on day $l$ for setting its consumption on that day to be equal to $\widehat{d}_i(l)-\pi^*_{\widehat{\theta}_\otimes,i}(z(l),\boldsymbol{\widehat{\delta}}(l)).$ The payment is determined as a function of all information that is available to the ISO until day $l$, namely, $\widehat{\theta}_\otimes$ and $\widehat{\boldsymbol{\delta}}^l$. 

The utility that DR load $i$ accrues on day $l$ is defined as $$u_{i,l}(S_i,\mathbf{S}_{-i},\theta_\otimes,\boldsymbol{\delta}^\infty)=p_{i,l}-c_i\big(\pi_{\widehat{\theta}_\otimes,i}^*(z(l),\boldsymbol{\widehat{\delta}}(l)),\delta_i(l)\big)$$ and its long-term average utility is defined as 
\begin{align}
    u_i^\infty(S_i,\mathbf{S}_{-i},\theta_\otimes,\boldsymbol{\delta}^\infty)\coloneqq\liminf_{L\to\infty}\frac{1}{L}\sum_{l=1}^Lu_{i,l}(S_i,\mathbf{S}_{-i},\theta_\otimes,\boldsymbol{\delta}^\infty).\label{longtermUtility}
\end{align}
Note that the above utility is a function not only of the strategy $S_i$ that DR load $i$ employs, but also of the strategies $\mathbf{S}_{-i}$ that the other DR loads employ. We will return to this point in the next subsection. Before that, it is necessary to recall the notion of a \emph{non-bankrupting strategy} that was recently introduced in \cite{arxivAER}.

\begin{definition}
A strategy $S_i$ of DR load $i,$ $i\in\{1,\hdots,n\},$ is \emph{non-bankrupting} if for all $(\mathbf{S}_{-i},\theta_\otimes),$
\begin{align}
    u_i^\infty(S_i,\mathbf{S}_{-i},\theta_\otimes,\boldsymbol{\delta}^\infty)>-\infty
\end{align}
almost surely. 

A strategy profile $\mathbf{S}=(S_1,\hdots,S_n)$ is non-bankrupting if $S_i$ is non-bankrupting for all $i\in\{1,\hdots,n\}.$
\end{definition}
{To elaborate, note that when the cost function of a DR load is bounded, the only way for its long-term average utility to be $-\infty$ is for its long-term average payment to be $-\infty,$ i.e., it pays an infinite sum to the ISO on average. Hence, the above definition essentially states that a DR load's strategy is non-bankrupting if it is guaranteed to not expend an infinite amount by employing that strategy, regardless of what strategies the other DR loads employ and what their type distributions are. }

\subsection{The Mechanism Design Problem}\label{sec_mechanismDesign}

There are three problems that the ISO faces in operating the grid optimally. The first is that the optimal day-ahead decision $(g^*_{\theta_\otimes},\boldsymbol{\pi}^*_{\theta_\otimes})$ is a function of the type distributions $\theta_1,\hdots,\theta_n$ which may not be reported truthfully. Secondly, even if the ISO were to somehow compute the optimal day-ahead decisions, the optimal real-time curtailment $\boldsymbol{\pi}^*_{{\theta}_\otimes}(z(l),\boldsymbol{\delta}(l))$ on any day $l$ is a function of the type realizations $\delta_1(l),\hdots,\delta_n(l)$ which may not be reported truthfully. Finally, the ISO cannot verify if a DR load $i$ faithfully reduces its consumption by $\pi_{\widehat{\theta}_\otimes,i}^*(z(l),\widehat{\boldsymbol{\delta}}(l))$ since it can only observe its actual consumption $y_i(l)$ and not its counterfactual baseline $d_i(l).$ 

All of these problems disappear if each DR load $i$ employs a truthful strategy. However, note from (\ref{longtermUtility}) that the utility of a DR load $i$ is a function not only of the strategy $S_i$ that it employs, but also of the strategies $\mathbf{S}_{-i}$ that the other DR loads employ. Consequently, a DR load $i$ may not employ a truthful strategy if there exists $(\mathbf{S}_{-i},\theta_\otimes)$ such that with some non-zero probability, DR load $i$ accrues a larger long-term average utility by employing a non-truthful strategy. This brings us to the mechanism design problem. We wish to design the payment rule $\{p_{i,l}:(i,l)\in\{1,\hdots,n\}\times\mathbb{N}\}$ such that each DR load $i$'s utility (\ref{longtermUtility}) is almost surely maximized by choosing $S_i\in\mathcal{T}_i$ regardless of what non-bankrupting strategy profile $\boldsymbol{S}_{-i}$ the other DR loads employ, and regardless of what ${\theta}_\otimes$ is. The next section presents such a payment rule.




\section{An Efficient and Incentive-Compatible Mechanism for Demand Response Markets}\label{mechanism}

For any $(i,l)\in\{1,\hdots,n\}\times\mathbb{N},$ the payment function $p_{i,l}$ consists of two components: (i) a first-stage payment $p_{i,l}^1$ that is determined based on only the type distributions reported in the day-ahead market and (ii) a second-stage settlement $p_{i,l}^2$ that is determined at the end of day $l$ based on the history of type realizations reported until day $l.$ We describe these payment rules next. They are an adaptation of the payment rule developed in \cite{arxivAER} for two-stage repeated stochastic games. 

\subsection{First-stage payment}
For every DR load $i,$ $i\in\{1,\hdots,n\},$ the first stage payment that it receives on any day $l$ is the Vickrey-Clarke-Groves (VCG) payment defined as 
\begin{align*}
    p_{i,l}^1(\widehat{\theta}_{\otimes})\coloneqq &W^*(\widehat{\theta}_{\otimes,-i})\nonumber\\
    -&\bigg[W^*(\widehat{\theta}_\otimes)-\mathbb{E}_{(z,\boldsymbol{\widehat{\delta}})\sim\theta_z\times\widehat{\theta}_{\otimes}}\big[c_i(\pi^*_{\widehat{\theta}_\otimes,i}(z,\widehat{\boldsymbol{\delta}}),\widehat{\delta}_i)\big]\bigg]
\end{align*}
where $W^*(\widehat{\theta}_{\otimes,-i})$ denotes the optimal social cost that would be attained if DR load $i$ were absent. 

\subsection{Second-stage settlement}
One of the primary functions of the second-stage settlement is to penalize DR loads for type bids whose empirical distributions are not consistent with the type distribution that they bid in the day-ahead market. Towards this, define for each $i\in\{1,\hdots,n\},$ $\nu\in\Delta$ and $L\in\mathbb{N}$, the empirical deviation
\begin{align}
    {f}_{i,\nu}(L)\coloneqq\bigg[\frac{1}{L}\sum_{l=1}^L{\bf 1}_{\{\widehat{\delta}_i(l)=\nu\}}\bigg]-\widehat{\theta}_i(\nu)
\end{align}
where $\widehat{\theta}_i(\nu)$ denotes the probability that a random variable distributed according to $\widehat{\theta}_i$ takes the value $\nu$.
It is easy to see that if DR load $i$ employs a truthful strategy, then $f_{i,\nu}(L)\to0$ almost surely for all $\nu\in\Delta$. 

Similarly, for every $i\in\{1,\hdots,n\},$ $\nu\in\Delta,$ $\boldsymbol{\eta}\in\Delta^{n-1},$ and $L\in\mathbb{N},$ define
\begin{align}
    h_{i,\nu,\boldsymbol{\eta}}(L)\coloneqq\bigg[\frac{1}{L}\sum_{l=1}^L&{\bf 1}_{\{\widehat{\delta}_i(l)=\nu,\widehat{\boldsymbol{\delta}}_{-i}(l)=\boldsymbol{\eta}\}}\bigg]\nonumber\\
    &-\bigg[\widehat{\theta}_i(\nu)\bigg]\bigg[\frac{1}{L}\sum_{l=1}^L{\bf 1}_{\{\widehat{\boldsymbol{\delta}}_{-i}(l)=\boldsymbol{\eta}\}}\bigg]
\end{align}
and note that $h_{i,\nu,\boldsymbol{\eta}}(L)\to0$ almost surely if the DR loads employ a truthful strategy. 

The second stage settlement rule checks on each day $L$ if $f_{i,\nu}(L)$ or $h_{i,\nu,\boldsymbol{\eta}}(L)$ exceeds a certain threshold $r(L)$ for some $(\nu,\boldsymbol{\eta})$ and imposes penalty $J_p(L)$ if one of them does. Towards this, define 
\begin{align*}
    E_{i}(l)\coloneqq\{\sup_{\nu}\vert f_{i,\nu}(l)\vert \geq r(l) \cup \sup_{(\nu,\boldsymbol{\eta})}\vert h_{i,\nu,\boldsymbol{\eta}}(l)\vert\geq r(l)\}.\label{EiDefn}
\end{align*}

How should the threshold sequence $\{r\}$ and the penalty sequence $\{J_p\}$ be chosen? The sequence $\{r\}$ should be chosen to balance two competing objectives. On the one hand, $r(l)$ must tend to $0$ as $l\to\infty$ since otherwise, the set of type bids that fall within the threshold will be ``large," thereby violating incentive compatibility. However, if the sequence shrinks to $0$ too quickly, then even truthful type bids may fall outside the threshold often, thereby incurring penalties often which in turn violates individual rationality. To balance the objectives, $\{r\}$ must be chosen such that $$\lim_{l\to\infty}r(l)=0,$$ and for some $\gamma>0,$ $$r(l)\geq\sqrt{\frac{\ln{2l^{1+\gamma}}}{2l}}.$$ The penalty sequence must be chosen to satisfy $$\lim_{l\to\infty}\frac{J_p(l)}{l}=\infty.$$
{See \cite[Section III]{arxivAER} for further intuition for these conditions. }

The second stage settlement that DR load $i$, $i\in\{1,\hdots,n\},$ receives on any day $l$, $l\in\mathbb{N},$ is defined as 
\begin{align}
    p_{i,l}^2(\widehat{\theta}_{\otimes},\widehat{\boldsymbol{\delta}}^l)\coloneqq\bigg[&c_i(\pi^*_{\widehat{\theta}_\otimes,i}(z(l),\widehat{\boldsymbol{\delta}}(l)),\widehat{\delta}_i(l))\nonumber\\
    &-\mathbb{E}_{(z,\boldsymbol{\widehat{\delta}})\sim\theta_z\times\widehat{\theta}_{\otimes}}[c_i(\pi^*_{\widehat{\theta}_\otimes,i}(z,\widehat{\boldsymbol{\delta}}),\widehat{\delta}_i)]\bigg]\nonumber\\
    &\;\;\;\;\;\;\;\;\;\;\;\;\;\;\;\;\;\;\;\;\;\;\;\;\;\;\;\;\;\;\;\;\;\;-J_p(l){\bf 1}_{E_i(l)}.
\end{align}
The total payment received by DR load $i$ on day $l$ is the sum of the first-stage payment and the second-stage settlement. I.e., 
\begin{align}
    p_{i,l}(\widehat{\theta}_{\otimes},\widehat{\boldsymbol{\delta}}^l)=p_{i,l}^1(\widehat{\theta}_{\otimes})+p_{i,l}^2(\widehat{\theta}_{\otimes},\widehat{\boldsymbol{\delta}}^l).\label{payment}
\end{align}
We now state the main result of the paper.

\begin{theorem}
Consider the mechanism defined by the decision rule (\ref{DAdecision}) and the payment rule (\ref{payment}). 
\begin{enumerate}
    \item For every $i\in\{1,\hdots,n\},$ $S_i\in\mathcal{S}_i$, $T_i\in\mathcal{T}_i,$ non-bankrupting strategy profile $\mathbf{S}_{-i},$ and $\theta_\otimes,$
    \begin{align}
        u_i^\infty(T_i,\mathbf{S}_{-i},\theta_\otimes,\boldsymbol{\delta}^\infty)\geq u_i^\infty(S_i,\mathbf{S}_{-i},\theta_\otimes,\boldsymbol{\delta}^\infty)
    \end{align}
    almost surely. 
    
    I.e., every DR load $i$ accrues a larger utility by employing a truthful strategy than by employing any other strategy, regardless of what (non-bankrupting) strategies the other DR loads employ.
    
    \item Suppose that $W^*({\theta}_{\otimes,-i})-W^*({\theta}_{\otimes})\geq0$ for all $\theta_\otimes.$ For every $i\in\{1,\hdots,n\},$ $T_i\in\mathcal{T}_i,$ $\mathbf{S}_{-i}\in\mathcal{S}_{-i},$ and $\theta_\otimes,$
    \begin{align}
        u_i^\infty(T_i,\mathbf{S}_{-i},\theta_\otimes,\boldsymbol{\delta}^\infty)\geq 0
    \end{align}
    almost surely.
    
    I.e., every DR load $i$ accrues a nonzero utility by employing a truthful strategy regardless of the strategies that the other DR loads employ.
    
    \item If for all $i\in\{1,\hdots,n\},$ $S_i\in\mathcal{T}_i,$ then, 
    \begin{align}
        \lim_{L\to\infty}\frac{1}{L}\sum_{l=1}^L\bigg[&c_g(g^*_{\widehat{\theta}_\otimes})+c_r(g_r(l))\nonumber\\
        +\sum_{j=1}^nc_j(&\pi^*_{\widehat{\theta}_\otimes,j}(z(l),\boldsymbol{\widehat{\delta}}(l)),\delta_j(l))\bigg]=W^*(\theta_\otimes)
    \end{align}
    almost surely.
    
    I.e., if all DR loads employ a truthful strategy, then the long-term average social cost that is incurred is almost surely equal to its optimal value.
\end{enumerate}
\end{theorem}

\begin{proof}
The result follows relatively straightforwardly from \cite[Theorem 1]{arxivAER} and is omitted in the interest of space. 
\end{proof}

\section{Numerical Results}\label{numerical}

In this section, we present a comparison of the proposed mechanism and the ``posted price mechanism" --- a popular alternative that has been employed in certain large-scale demand response trials in Europe. We first describe the posted price mechanism and then present simulation results which quantify the difference in the social cost between the posted price mechanism and the proposed mechanism. 

\subsection{The Posted Price Mechanism}

The posted price mechanism involves the ISO announcing a rebate for consumption reductions during times of power supply shortage. The DR loads react to the announced price and optimize their energy consumption which leads to energy consumption reductions. 


To elaborate, denote by $p$ the rebate that the ISO provides a DR load for every unit of reduction in its energy consumption. On each day $l,$ each DR load $i$ determines its consumption reduction $x^*_i(l)$ as
\begin{align}
    x^*_i(l)=\argmax_{x} \big[px-c_i(x,\delta_i(l))\big].
\end{align}
Consequently, the total demand-supply mismatch on day $l$ is $$z(l)-\sum_{i=1}^nx_i^*(l)=:g_r(l)$$ which the ISO must purchase from the spot market on day $l.$ This costs $c_r(g_r(l))$ on day $l.$ Consequently, the social cost incurred on day $l$ equals $$c_r\big(z(l)-\sum_{i=1}^nx_i^*(l)\big)+\sum_{i=1}^n c_i(x_i^*(l),\delta_i(l)).$$ 
Note that the above quantity is a function of the price $p$ offered by the ISO, and in Fig. \ref{simulationPlot}, the dotted curve plots the average social cost incurred by the posted price mechanism as a function of the price $p,$ averaged over $1000$ days. Section \ref{subsec_details} contains further details.

\subsection{The optimal mechanism}

The mechanism proposed in the Section \ref{mechanism} incentivizes each DR load to bid its type distribution and type realizations truthfully, thereby allowing the system operator to determine the optimal energy consumption reductions and reserve generation on each day $l$ as a solution to optimization program
\begin{align}
    \min_{x_1,\hdots,x_n} c_r(z(l)-\sum_{i=1}^n x_i)+\sum_{i=1}^n c_i(x_i,\delta_i(l)).
\end{align}
The optimal value of this program is a random variable that depends on the realization of $\boldsymbol{\delta}(l).$ The solid curve of Fig. \ref{simulationPlot} plots the optimal value averaged over a duration of $1000$ days. Note that the above quantity is independent of the price $p$ defined as a part of the posted price mechanism. The variation of the solid curve across different values of $p$ stems solely from averaging a random variable over a finite number of realizations. 

\subsection{Simulation parameters}\label{subsec_details}

For simulations, we assume a system with $10000$ demand response-providing loads, i.e., $n=10000$. We take all costs to be quadratic so that $c_i(x,\delta_i(l))=\frac{\delta_i(l)}{2}x^2$ for all $i,l$ and $c_r(x)=5x^2$. The distribution of $\delta_i(l)$ is taken to be uniform in $\{1,\hdots,10\}$ for all $i,l$, and the net demand of the inelastic loads $z(l)$ on each day $l$ is taken to be uniformly distributed in the interval [$0$J, $100$J]. We simulate both the posted price mechanism and the proposed mechanism for a duration of $1000$ days. 

An alert reader would have noticed that in the description of the posted price mechanism, we have implicitly assumed the ISO to know the baseline of each DR load --- an advantage not assumed for the ISO in the proposed mechanism. In spite of this advantage, the social cost incurred by the posted price mechanism is almost thrice as large as the social cost incurred by the proposed mechanism.

\begin{figure}
    \centering
    \includegraphics[width=\columnwidth]{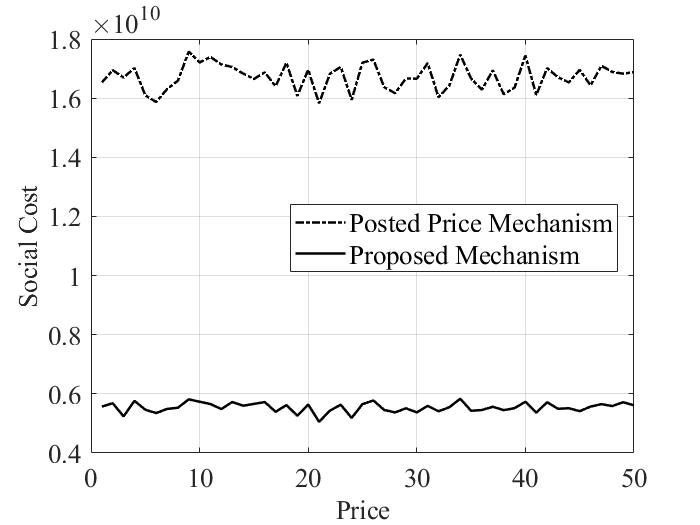}
    \vspace*{-7mm}
    \caption{Social cost incurred by the proposed mechanism and the posted price mechanism as a function of the rebate offered by the ISO.}
    \label{simulationPlot}
\end{figure}

\section{Conclusion}\label{conclusion}
A key difficulty in demand response schemes is to determine if a load indeed reduces its power consumption when called upon for demand response, and if so by how much. This is challenging because the system operator can only observe the actual consumption of a load and not its counterfactual baseline. Additionally, determining the optimal curtailment levels of the loads is also challenging as they depend on the loads’ privately-known costs which the loads may misreport. We have designed a two-stage mechanism to address these issues. The mechanism guarantees ex post incentive compatibility of truthful bidding in both the day-ahead market and in real time, ex post individual rationality, and attains the optimal social cost. 

\bibliographystyle{IEEEtran}
\bibliography{references.bib}

\appendix
\section*{List of Symbols}
\begin{tabular}{ll}
    $z(l)$ & Net demand of inelastic loads on day $l$\\
    $\theta_z$ & Probability distribution of $z(1)$\\
    $d_i(l)$ & Baseline of DR load $i$ on day $l$\\
    $\widehat{d}_i(l)$ & Reported baseline of DR load $i$ on day $l$\\
    $c_g$ & Production function of generator\\
    $c_r$ & Production function of reserve\\
    $c_i$ & Cost function of DR load $i$\\
    $\delta_i(l)$ & Type of DR load $i$ on day $l$\\
    $\widehat{\delta}_i(l)$ & Reported type of DR load $i$ on day $l$\\
    $\Delta$ & Type space of DR loads\\
    $\theta_i$ & Type distribution of DR load $i$\\
    $\widehat{\theta}_i$ & Reported type distribution of DR load $i$\\
    $\Theta$ & Set of type distributions of a DR load\\
    $\theta_\otimes$ & Joint distribution of DR loads' types\\
    $\widehat{\theta}_\otimes$ & Reported joint dist. of DR loads' types\\
    $\boldsymbol{\pi}$ & Demand response policy\\
    ${\pi}_i$ & DR policy for the $i$th DR load\\
    $\boldsymbol{\pi}^*_{{\theta_\otimes}}$ & Optimal DR policy for type distribution $\theta_\otimes$\\
    ${\pi}^*_{{\theta_\otimes},i}$ & Optimal DR policy for $i$th load for dist. $\theta_\otimes$\\
    $\boldsymbol{\pi}^*_{{\widehat{\theta}_\otimes}}$ & Optimal DR policy for type distribution $\widehat{\theta}_\otimes$\\
    ${\pi}^*_{{\widehat{\theta}_\otimes},i}$ & Optimal DR policy for $i$th load for dist. $\widehat{\theta}_\otimes$\\
    $g$ & Energy dispatch of generator\\
    $g^*_{{{\theta}_\otimes}}$ & Optimal energy dispatch for distribution ${\theta}_\otimes$\\
    $g^*_{{\widehat{\theta}_\otimes}}$ & Optimal energy dispatch for distribution $\widehat{\theta}_\otimes$\\
    $g_r(l)$ & Energy production of reserve on day $l$\\
    $W$ & Expected social cost function\\
    $W^*$ & Optimal expected social cost function\\
    $y_i(l)$ & Energy consumption of DR load $i$ on day $l$\\
    $\sigma_i$ & Day-ahead bidding strategy of DR load $i$\\
    $\phi_i$ & Real-time bidding policy of DR load $i$\\
    $\Phi_i$ & Set of real-time bidding policies of DR load $i$\\
    $\mu_i$ & Real-time bidding strategy of DR load $i$\\
    $S_i$ & Strategy of DR load $i$\\
    $\mathcal{S}_i$ & Set of strategies available to DR load $i$\\
    $\mathcal{T}_i$ & Set of truthful strategies of DR load $i$\\
    ${p}^1_{i,l}$ & First-stage payment of DR load $i$ on day $l$\\
    ${p}^2_{i,l}$ & Second-stage payment of DR load $i$ on day $l$\\
    ${p}_{i,l}$ & Total payment of DR load $i$ on day $l$\\
    ${u}_{i,l}$ & Utility accrued by DR load $i$ on day $l$\\
    ${u}^\infty_i$ & Long-term average utility accrued by load $i$\\
    ${r}(l)$ & Threshold on day $l$\\
    ${J_p}(l)$ & Penalty on day $l$\\
\end{tabular}

\end{document}